\documentclass{IOS-Book-Article}     

\usepackage{mathtools}

\DeclarePairedDelimiter\floor{\lfloor}{\rfloor}

\usepackage{hyperref,   amssymb,amsbsy,amsfonts,amsthm,latexsym,amsopn,amstext, amsxtra,euscript,amscd}

\usepackage{hyperref}

\hypersetup{
  colorlinks   = true, 
  urlcolor     = blue, 
  linkcolor    = blue, 
  citecolor   = red 
}

\usepackage{cite}

\usepackage[msc-links]{amsrefs}

\newtheorem{thm}{Theorem}
\newtheorem{lem}{Lemma}
\newtheorem{cor}{Corollary}
\newtheorem{prop}{Proposition}

\newtheorem{exmp}{Example}
\newtheorem{rem}{Remark}

\newtheorem{prob}{Problem}
\newtheorem{definition}{Definition}

\def\F{\mathbb F}
\def\Z{\mathbb Z}
\def\C{\mathbb C}
\def\P{\mathbb P}
\def\Q{\mathbb Q}

\def\X{\mathcal X}

\newcommand{\supp}{{\rm supp}}
\newcommand{\codeL}{{\mathcal L}}
\newcommand{\res}[2]{\mathrm{res}_{#1} \left( #2 \right)}

\DeclareMathOperator{\x}{\mbox{\bf x}}
\DeclareMathOperator{\y}{\mbox{\bf y}}
\DeclareMathOperator{\wt}{\mbox{wt}}

\begin{document}
\begin{frontmatter}          
%
\title{Weight distributions, zeta functions and Riemann hypothesis for linear and algebraic geometry codes}

\runningtitle{Zeta functions for linear codes}


%
\author[A]{\fnms{Artur} \snm{Elezi}}
\author[B]{\fnms{Tony} \snm{Shaska}}
\address[A]{Department of Mathematics, \\American  University, \\Washington DC, 20016}
\address[B]{Department of Mathematics, \\Oakland University, \\Rochester, MI, 48309.}

\maketitle

\begin{abstract}
This is a survey on weight enumerators, zeta functions and Riemann hypothesis for linear and algebraic-geometry codes.
\end{abstract}


\end{frontmatter}



\section{Introduction}

For more than $150$ years, generations of mathematicians have been mesmerized by and hard at work to muster/solve the Riemann zeta function and Riemann Hypothesis-Weil Conjectures in various dimensions. The classic one, over $Spec(\Z)$ is still unsolved. A. Weil successfully completed the task for curves over finite fields $\F_q$ around mid $19$-th century. Amazingly, in the past 15-20 years or so yet another context has been provided for Riemann zeta function and Riemann hypothesis: linear codes! While the weight distributions/enumerators of linear codes are important in themselves, they give rise to analogous zeta functions and Riemann Hypothesis. The connections between these different settings are beautiful. The goal of this survey is to provide a short and gentle introduction to zeta functions and the Riemann hypothesis for linear codes.

This survey is organized as follows:  In section two we introduce the basics of linear codes, and their their weight enumerators. Next, we provide an elementary proof of MacWilliams' identity for dual codes. Finally we discuss general solutions of MacWilliams' equations, and as a special case obtain the weight enumerator of an MDS code.  

In section three, we introduce and provide some historical background and motivation for zeta functions of linear codes. Various functional identities for zeta polynomials and zeta functions have been provided. Next, Riemann Hypothesis is introduced for general virtual and formal weight enumerators - a straight forward generalization of weight enumerators. Last, a discussion on (virtual) codes that satisfy Riemann Hypothesis follows.

in section four, we introduce generalities of algebraic curves and algebraic geometry (AG) codes that arise from them. Following Duursma \cite{D6}, determining the weight distribution of AG codes has been reformulated and discussed as the problem of the effective divisors distributions over divisor classes where we can take advantage of the group structure. Complete results follow for rational and elliptic curves. 

We have tried to address an audience of beginners, especially graduate students. From this point of view, we have selected proofs that are accessible, and whenever possible rather elementary. To the extent possible, the survey is self-contained and a few open problems have been presented. 


\bigskip

\noindent \textbf{Notations:}  Throughout this paper $\F_q$ will denote a field of $q$ elements. By a curve $\X_g$ we denote an irreducible, projective, smooth algebraic curve of genus $g$ over $\F_q$. The set of rational points of $\X_g$ over $\F_{q^r}$ will be denoted by $\X_g(\F_{q^r})$. A linear code will be denoted by $C$.  The cardinality of a set $S$ will be denoted by $\# S$.


\section{Codes and their weight enumerators}

\subsection{Codes} 
Let $q$ be a prime power. A $q$-ary code $C$ of length $n$  is a subset of $\F_q^n$.  Elements of $C$ are called \textit{codewords}, those of $\F_q^n$ are called \textit{words}. The \textit{Hamming distance} between ${\x}=(x_1, \dots ,x_n)$ and $\y =(y_1, \dots ,y_n)$ is defined as  
\[ d( \x, \y):=\#\{i: x_i\neq y_i\}.\]
The smallest of the distances between distinct codewords is called \textit{the minimum distance} of the code $C$. The \textit{weight} of a word ${\x}=(x_1, \dots ,x_n)$ is defined as 
\[\wt( \x ):=\#\{i: x_i\neq 0\}.\]
A code $C$ is called \textit{linear} if it is a linear subspace of $\F_q^n$. For such a code, \textit{the minimum distance equals the smallest of the weights of nonzero codewords} of $C$. Fix a basis $\{ r_1, r_2, \dots r_k\}$ of $C$. The maximal rank $k\times n$ matrix $\bf G$ whose rows are $r_1, r_2, \dots r_k$ is called the \textit{generator matrix} of $C$. Each codeword might be identified with a vector 
\[ \x  = (x_1,x_2, \dots ,x_k)\in \F^k_q.\]
\textit{Encoding} $\x$ via the code $C$ means multiplying it to the right by $\bf G$, i.e $\x$ is encoded to 
\[ {\x}{\bf G}=x_1 r_1+ \dots +x_k r_k,\]
which is an element of $C$. Instead of the $k$-string $\x$, the encoded string ${\x}{\bf G}$ is transmitted. The quantity $k/n$ is called \textit{the rate} of the linear code $C$. 

Let ${\bf a} \cdot {\bf b} =a_1 b_1+ \dots +a_n b_n$ be the standard dot product in $\F_q^n$. The dual of $C$ is defined as 
\[ C^{\perp}:=\{ \x    \in \F_q^n~|~ {\x \y} =0~ \text{for all} ~\y\in C\}.\]
It is a linear code of length $n$ and dimension $n-k$. The generator matrix $\bf H$ of $C^{\perp}$ is called \textit{the parity check} matrix of $C$. It has dimensions $(n-k)\times n$ and satisfies 
\[ C= \{   \x \in \F_q^n~:~{\bf H}   \x =  0\}.\]
A linear code $C$ is called \textit{self-orthogonal} if and only if  $C\subset C^{\perp}$. It is called \textit{self-dual} if and only if  $C= C^{\perp}$. Self-dual codes that arise from algebraic geometry constructions, are of special importance for various reasons, one of them being their use in quantum computing (\cite{E}, \cite{ES1}, \cite{ES2}). 

If a codeword $\x=(x_1, \dots x_n)$ is sent and a word $\y=(y_1, \dots ,y_n)$ is received, the error made during transmission is the word $(e_1, \dots ,e_n)=e:=\y-\x$. Notice that %
\[ d( \x, \y)=\#\{i: x_i\neq y_i\}=\#\{i: e_i\neq 0\}=\wt(e).\]
Nearest neighbor \textit{decoding} of a received word $\y$ is the "closest" codeword, i.e. the codeword $\x$ such that the Hamming distance $d(\x,\y)$ (or alternatively the error weight $\wt(e)$) is minimum. Note that $\x$ may not be unique. If $d-1$ or fewer errors are made in transmitting a codeword $\x\in C$, the received word $\y$ is no longer in $C$ (otherwise $d(\x, \y)\leq d-1$). Hence the receiver knows that errors have occurred during this transmission. It is said that $C$ \textit{detects} up to $d-1$ errors. For each word $\y\in \F_q^n$, there is only one codeword $\x\in C$ of distance up to $\left\lfloor{(d-1)/2}\right\rfloor$. Indeed, if there were two, the distance between them would be less than $d$ by the triangle inequality. It follows that nearest neighbor decoding is successful if up to  $\left\lfloor{(d-1)/2}\right\rfloor$  errors have occurred. It is said that $C$ \textit{corrects} up to $\left\lfloor{(d-1)/2}\right\rfloor$  errors.

A $q$-ary linear code of length $n$, dimension $k$ and minimum distance $d$ is denoted by $[n,k,d]_q$. As established above, nearest neighbor decoding detects up to $d-1$ errors and corrects up to $\left\lfloor{(d-1)/2}\right\rfloor$ errors.

\subsection{Distance and weight enumerators} For a subset $S\subset \F_q^n$ and $0\leq i\leq n$, let 
\[ A_i:=\#\{{\bf c}\in S~:~\wt({\bf c})=i\},~B_i:=\displaystyle{\frac{1}{\#(S)}}\#\{({\bf c}_1,{\bf c}_2)\in S\times S~:~d({\bf c}_1,{\bf c}_2)=i\}.\]
The vectors $(A_0,\ldots,A_n)$ and $(B_0,\ldots,B_n)$ are called respectively the \textit{weight distribution} and \textit{the distance distribution} of $S$. 

\begin{definition}

(a) The \textit{Hamming weight enumerator} of $C$ is the generating function 
\[ W_S(z):=\sum_{i=0}^n A_iz^i \]
or its homogenization 
\[A_S(x,y):=x^nW_S(y/x)=\sum_{i=0}^n{A_ix^{n-i}y^i}=\sum_{c\in S}x^{n-\wt(c)}y^{\wt(c)}\in \Z[x,y] .\]

(b) The \textit{Hamming distance enumerator} of $C$ is the generating function  
\[ D_S(z)=\sum_{i=0}^n B_iz^i \]
or its homogenization
\[
\begin{split}
B_S(x,y) & :=x^nD_S(y/x)=\sum_{i=0}^n{B_ix^{n-i}y^i}\\
& =\sum_{c_1,c_2\in S}x^{n-d(c_1,c_2)}y^{d(c_1,c_2)}\in \Q [x,y]. 
\end{split}
\] 
\end{definition} 
\noindent For linear codes the two notions coincide.

\begin{prop}
If $C$ is a linear code then $W_C(z)=D_C(z)$.
\end{prop}

\begin{proof}
Notice that $d({\x},\y)=\wt(\bf{x}-\y).$ Hence, if $\wt({\bf a})=i$ then 
\[ \forall {\bf c}\in C,~d({\bf c},{\bf a+c})=i.\]
It follows that 
\[ \#\{({\x,y})\in C\times C~:~d({\bf c_1,c_2})=i\}=\#(C)A_i \] 
Now the proposition follows easily.
\end{proof}

\begin{exmp}
Consider the repetition code $i_2=\{00,11\}$. It is a binary self-dual code with weight enumerator 
\[ A_{i_2}(x,y)=x^2+y^2.\]
\end{exmp}

\begin{exmp}
The $[4,2,3]_3$ tetra code $t_4$ generated by $\{1110,0121\}$ has 
\[ A_{t_4}(x,y)=x^4+8xy^3.\]
\end{exmp}

 \noindent Note that the weight enumerator of an $[n,k,d]_q$-code 
 \[ A_C(x,y)=x^n+\sum_{i=d}^nA_ix^{n-i}y^i\]
depends only on the parameters $n,d$ and not on $q$. 
  
\begin{definition}
 Two codes are said to be \textit{formally equivalent} if they have the same weight distribution. 
\end{definition}
 
\noindent The following are natural problems in coding theory. 

\begin{prob} 
Given the weight enumerator $A(x, y)$ of a code, how many non-equivalent codes are there corresponding to  $A(x, y)$?  
\end{prob}

\begin{prob}
Given a homogeneous polynomial with nonnegative integer coefficients 
\[ F(x,y)=x^n+\sum_{i=1}^{n} f_ix^{n-i}y^i,\]
under what conditions does there exists a linear code $C$ such that $A_{C}(x,y)=F(x,y)$?
\end{prob}

For examples of non-equivalent codes corresponding to the same weight enumerator the reader can check \cite{s-shor},\cite{cod-1}, and \cite{cod-2}.

\subsection{Dual codes and their weight enumerators.} The weight distribution of the dual $C^{\perp}$ can be recovered from the weight distribution of $C$ by applying a linear transformation. 

\begin{thm} [MacWilliams' Identity] 
For an $[n,k,d]_q$-code $C$
\[W_{C^{\perp}}(z)=\frac{[1+(q-1)z]^n}{q^k}W_C\left(\frac{1-z}{1+(q-1)z}\right)=\frac{1}{q^k}\sum_{i=0}^nA_i[1+(q-1)z]^{n-i}(1-z)^i.\]
\noindent Equivalently
\[ A_{C^{\perp}} (x,y)=\frac{1} {q^k}   \, A_C \left(x+(q-1)y, x-y \right) \]
\end{thm}

\begin{proof}  There are many proofs of this theorem, we present an elementary approach (see \cite{H}). For any $S\subset \F^n_q$, define 
\[ U_S(z):=[1+(q-1)z]^n W_S\left(\frac{1-z}{1+(q-1)z}\right).\]
It is often called the \textit{MacWilliams' transform} of $W_C(z)$. With this definition, MacWilliams' Identity reads: 
\[ W_{C^{\perp}}(z)=\frac{1}{q^k}U_C(z).\]
First, some preliminaries. Note that if $S,T$ are disjoint then 
\[ W_{S\cup T}(z)=W_S(z)+W_T(z)~\text{and}~U_{S\cup T}(z)=U_S(z)+U_T(z).\]
If a code $C$ is \textit{decomposable}, i.e. if up to a permutation of coordinates, $C$ is a direct product $C_1\times C_2$ of linear codes with positive length, 
then $C^{\perp}=C_!^{\perp}\times C_2^{\perp}$ and 
\begin{equation} \label{multiplicativity}
W_C(z)=W_{C_1}(z)W_{C_2}(z),~~U_C(z)=U_{C_1}(z)U_{C_2}(z).
\end{equation}
Back to the proof of MacWilliams' Identity.  We use induction on the length $n$ of the code $C$. For $n=1$ there are two cases: 

(a) $C=\{0\},~C^{\perp}=\F_q,~W_C(z)=1,~W_{C^{\perp}}(z)=1+(q-1)z$. 

(b) $C=\F_q,~C^{\perp}=\{0\},~W_C(z)=1+(q-1)z,~W_{C^{\perp}}(z)=1$.

In each of these cases, MacWilliams' Identity is easily verified directly. For example, in case (b): 
\[ \frac{1}{q^k}U_C(z)=\frac{1}{q}(1+(q-1)z)\left(1+(q-1)\frac{1-z}{1+(q-1)z}\right)=1=W_{C^{\perp}}(z).\]

Assume that MacWilliams' Identity hold for codes of length less than $n>1$ and let $C$ be a code of length $n$. If $C$ is decomposable, the assertion follows from the induction hypothesis and the multiplicativity Eq.~\eqref{multiplicativity} of $W_C$ and $U_C$. If $C$ is indecomposable, neither $C$ nor $C^{\perp}$ contains a word of weight $1$. Let \break $C_0=\{{\bf c}\in C~|~c_n=0\},~C_1=C-C_0$ and 
\begin{equation} \label{nonzeroC}
{\bf a}:=\{a_1, \dots a_{n-1},1)=({\bf a}',1)\in C
\end{equation}
such that $C=C_0\oplus \F_q\cdot {\bf a}$. The projection map 
\[ \pi:\F_q^n\rightarrow \F_q^{n-1},~\pi(x_1,\ldots,x_{n-1},x_n)=(x_1,\ldots,x_{n-1})\]
is injective on $C$, otherwise there will be two elements of $C$ whose Hamming distance is $1$. It follows that $\pi(C)$ is a disjoint union of $\pi(C_0$ and $\pi(C_1)$, hence 
\begin{equation}
W_{\pi(C)}(z)=W_{\pi(C_0)}(z)+W_{\pi(C_1)}(z)=W_{\pi(C_0)}(z)+\frac{W_{C_1}(z)}{z}
\end{equation}
\noindent and in turn
\begin{equation} \label{splittingU}
\begin{split}
U_{\pi(C)}(z) & = U_{\pi(C_0)} (z) + \frac{ \left[1+(q-1)z \right]^n} {1-z} \, W_{C_1}  \left(  \frac{1-z} {1+(q-1)z}  \right) \\
& = U_{\pi(C_0)}(z)+(1-z)^{-1}U_{\pi(C_1)}(z). \\
\end{split}
\end{equation}
\noindent Assume that $(b_1,b_2, \dots ,b_{n-1},b_n)\in C^{\perp}$. If $b_n=0$ then $(b_1, \dots ,b_{n-1})\in \pi(C)^{\perp}$. If $b_n\neq 0$ then 
\begin{itemize}
\item For any $(c_1, \dots c_{n-1},0)\in C_0$ we have $b_1c_1+ \dots +b_{n-1}c_{n-1}=0$. It follows that ${\bf b}':=(b_1,b_2, \dots ,b_{n-1})\in {\pi(C_0)}^{\perp}$.
\item $a_1b_1+ \dots +a_{n-1}b_{n-1}+b_n=0$. Recall from Eq.~\eqref{nonzeroC} that ${\bf a}'=(a_1,a_2, \dots ,a_{n-1})$ and let ${\bf a'b'}$ is the standard dot product in $\F_q^{n-1}$. Then $b_n=-{\bf a'b'}$. Since $b_n\neq 0$ and ${\bf a}'=\pi({\bf a})\in \pi(C)$. It follows that ${\bf b}'\notin \pi(C)^{\perp}$.
\end{itemize}

\noindent It follows that $C^{\perp}$ is a disjoint union of $C'_0:=\{({\bf b}',0)~|~{\bf b}'\in {\pi(C)}^{\perp}\}$ and 
\[ C'_1:=\{({\bf b}',-{\bf a'b'})~|~{\bf b}'\in\pi(C_0)^{\perp}-\pi(C)^{\perp}\}\]
and therefore 
\begin{equation}
\begin{split}
W_{C^\perp}(z) & =W_{\pi(C)^\perp}(z) + z \left(W_{\pi(C_0)^\perp}(z)-W_{\pi(C)^\perp}(z) \right)\\
& =(1-z)W_{\pi(C)^\perp}(z)+zW_{\pi(C_0)^\perp}(z)
\end{split}
\end{equation}
\noindent Notice that $\text{dim}~\pi(C)=k~\text{and}~\text{dim}~\pi(C_0)=k-1$. Applying the inductive hypothesis in the last identity we obtain
\[W_{C^\perp}(z)=\frac{1-z}{q^k}U_{\pi(C)}(z)+\frac{z}{q^{k-1}}U_{\pi(C_0)}(z).\]
We now use identity \eqref{splittingU} and get 
\[
\begin{split}
W_{C^\perp}(z) & = \frac{1}{q^k}  \left(   \left[1+(q-1)z\right]     U_{\pi(C_0)}(z)+U_{C_1}(z)\right) \\
& =\frac{1}{q^k}(U_{C_0}(z)+U_{C_1}(z))=\frac{1}{q^k}U_C(z).
\end{split}
\]
as desired. 
\end{proof}

\begin{exmp}
The binary repetition code $i_2=\{00,11\}$ is self-dual. Its weight enumerator $A(x,y)=x^2+y^2$ is left unchanged when $x,y$ are replaced by 
\[ \frac {x+y} {\sqrt 2}   \quad \textit{ and } \quad \frac {x-y} {\sqrt 2} .\]
\end{exmp}

\begin{exmp}
The repetition code $C$ over $\F_q$ has weight enumerator 
\[ A_C(x,y)=x^n+(q-1)y^n.\]
Its dual code has weight enumerator 
\[ A_{C^\perp}(x,y)=\frac{1}{q}[(x+(q-1)y)^n+(q-1)(x-y)^n].\]
Note that $A_C=A_{C^\perp}$ when $n=2$.
\end{exmp}

\begin{definition}
A linear code is said to be \textit{formally self-dual} if 
\[ A_C(x,y)=A_{C^{\perp}}(x,y).\]
\end{definition}

\begin{exmp}  Let $C$ be the binary code generated by
 \begin{equation}
G=\left(\begin{array}{cccccccccc}
1 & 1 & 1 & 1 & 0 & 0  & 0 & 0 & 0 & 0\\
0 & 0 & 0 & 0 & 1 & 1  & 1 & 1 & 1 & 1 \\
1 & 0 & 0 & 0 & 1 & 1  & 1 & 0 & 0 & 0\\
0 & 1 & 0 & 0 & 1 & 1  & 0 & 1 & 0 & 0\\
0 & 0 & 1 & 0 & 1 & 0  & 1 & 0 & 1 & 0\\

\end{array} \right).
\end{equation}
The weight distributions of $C$ and $C^\perp$ are the same: 
\[ (1,0,0,0,15,0,15,0,0,0,1),\]
yet, $C\neq C^\perp$. So, this code is formally self-dual but not self-dual.
\end{exmp}

\subsection{MDS codes and their weight enumerators} Let $C$ be an $[n,k,d]_q$ code. By the Singleton's bound, $d\leq n+1-k$. The dual code $C^{\perp}$ has parameters $[n,n-k,d^{\perp}]$ with $d^{\perp}\leq k+1$. 

\begin{definition} The genus of an $[n,k,d]_q$-code $C$ is defined by 
\[ \gamma(C)=n+1-k-d.\]
\end{definition}

\noindent Notice that for a self-dual code $C$, its length is even and its dimension is $n/2$, hence 
\[\gamma(C)=n/2+1-d.\]

\begin{definition}
An $[n,k,d]_q$ code $C$ is called MDS (maximum distance separating) if and only if  its genus is $0$, i.e. if and only if  it achieves its Singleton bound.
\end{definition}

\noindent In light of the above definition, the genus measures how far the code is from being MDS. It is well known that if there exists an MDS code with parameters $[n,k,n-k+1]_q$ then $n\leq q+k-1$. 

\begin{prop}
A code $C$ is MDS iff  $C^{\perp}$ is MDS, i.e. $\gamma(C)=0$ iff $\gamma(C^\perp)=0$.
\end{prop}

\begin{proof}
Let $C$ be an MDS code of dimension $k$ and minimum distance $d=n-k+1$. The dimension of its dual $C^{\perp}$ is $n-k$. Let $d^{\perp}$ denote the minimum distance of $C^{\perp}$. By the Singleton bound, $d^{\perp}\leq n-(n-k)+1=k+1$. We will show that $d^{\perp}\geq k+1$. Assume by way of contradiction that there is a word ${\bf c}\in C^{\perp}$ with weight at most $k$. Without loss of generality we assume that ${\bf c}=(c_1, \dots ,c_k,0,0, \dots ,0)$. It follows that for every word ${\bf b}\in C$ we have 
\begin{equation}\label{lineareq}
c_1b_1+ \dots +c_kb_k=0.
\end{equation}
\noindent Let $\pi:\F_q^n\rightarrow \F_q^k$ be the projection map that "forgets" the last $n-k$ coordinates. Since the minimum distance of $C$ is $n-k+1$, the map $\pi$ is an isomorphism of $C$ onto $\F_q^k$. But then, Eq.~\eqref{lineareq} represents a non degenerate linear form that vanishes on $\F_q^k$. This is a contradiction.
\end{proof}

\noindent Notice that 
\[\gamma(C)+\gamma(C^\perp)= n+2-d-d^\perp.\]
\noindent It follows from the proposition that $d+d^\perp=n+2$ iff both $C$ and $C^\perp$ are MDS, and $d+d^\perp\leq n$ iff none of them is. For a linear code $C$, $d+d^\perp\neq n+1$.

\begin{thm} \label{MacWilliams solutions} (Solutions of MacWilliams equations.)
Let $C$ be a linear code of length $n$ and minimum distance $d$. Let $d^\perp$ be the minimum distance of its dual $C^\perp$. If $C$ is MDS, then its weight distribution is
\[A_0=1,~\text{and}~A_i={n \choose i}(q-1)\sum_{m=0}^{i-d}{i-1 \choose m}(-1)^m q^{i-d-m}, ~d=n-k+1\leq i\leq n.\]
Otherwise, its weight distribution is determined by $A_d,A_{d+1},\ldots A_{n-d^\perp}$.
\end{thm}

\begin{proof} Let $(A_0^{\perp}, \dots ,A_n^{\perp})$ be the weight distribution of the dual code $C^{\perp}$. Using MacWilliams' Identity we may write 
\[ \sum_{i=0}^nA_i^{\perp}z^i=\frac{1}{q^k}\sum_{i=0}^nA_i[1+(q-1)z]^{n-i}(1-z)^i.\]
Multiplying both sides by $z^{-n}$ and substituting $\displaystyle{z=\frac{1}{1+t}}$ yields 
\[\sum_{i=0}^nA_i^{\perp}(1+t)^{n-i}=\frac{1}{q^k}\sum_{i=0}^nA_i(q+t)^{n-i}t^i.\]
Reverse the roles of $C$ and $C^{\perp}$ and compare the coefficients of powers of $t$ in both sides. We obtain

\begin{equation}\label{eq-8}
\sum_{i=0}^{n-l} {n-i \choose l}A_i=q^{k-l}\sum_{i=0}^l {n-i \choose n-l}A_i^{\perp},~0\leq l\leq n.
\end{equation}

\noindent Notice that $A_0=A^{\perp}_0=1$ and $A_1=\ldots=A_{d-1}=A^{\perp}_1=\dots =A^{\perp}_{d^\perp -1}=0$. Therefore, we get
\[\sum_{i=d}^{n-l}{n-i \choose l} A_i={n \choose l}(q^{k-l}-1),~l=0,\ldots d^{\perp}-1.\] 
\noindent This is a linear system with $d^\perp$ equations and $n+1-d$ unknowns $A_d,A_{d+1},\ldots A_n$. 

Case 1: $n+1-d^\perp =d-1$ or $d+d^\perp =n+2$. Both $C$ and $C^{\perp}$ are MDS codes of length $n$. The $l=d^\perp -1$ equation is a trivial identity. The remaining $d^\perp-1$ equations form a linear system in $n+1-d=d^\perp-1$ unknowns which can be solved iteratively. For $l=d^\perp -2$ we get 
\[ A_d={n \choose d}(q-1).\]
Substitute this into the $l=d^\perp-3$ equation to find $A_{d+1}$, and so on. We obtain
\begin{equation} \label{A_i}
A_i={n \choose i}(q-1)\sum_{m=0}^{i-d}{i-1 \choose m}(-1)^m q^{i-d-m}, ~d\leq i\leq n.
\end{equation}

Case 2: $n+1-d^\perp>d+1$ or $d+d^\perp <n+2$. Neither $C$ nor $C^\perp$ is MDS, therefore $d+d^\perp \leq n$. The linear system can be solved iteratively as follows. The last equation with $l=d^\perp -1$ is
\[\sum_{i=d}^{n+1-d^\perp}{n-i \choose d^\perp-1} A_i={n \choose d^\perp -1}(q^{k+1-d^\perp}-1),~l=0,\ldots d^{\perp}-1.\]
We get $A_{n+1-d^\perp}$ in terms of $A_d,\ldots,A_{n-d^\perp}$. Substituting in the next to last equation we obtain $A_{n+2-d^\perp}$, and so on. The weight distribution is determined in terms of $A_d,\ldots,A_{n-d^\perp}$.
\end{proof}


\section{Zeta Functions for Codes}

In this section we study  zeta function of a linear code. First, we discuss some history and motivation.

\subsection{Classic Riemann Zeta Function} In the middle of 19th century, Bernhard Riemann formulated the much important and yet unsolved \textit{Riemann Hypothesis} regarding the distribution of the zeros  of the Riemann zeta-function $\zeta(s)$. This function is defined via a series 
\[ \zeta(s)=\sum_{n=1}^{\infty}\frac{1}{n^s}\]
which is convergent for $\{s\in \C~:~Re(s)>1\}$. By analytic continuation, Riemann showed that $\zeta(s)$ extends to a meromorphic function on $\C$ with a simple pole at $s=1$ of residue one. It satisfies the functional equation 
\[ \pi^{-s/2}\Gamma(\frac{s}{2})\zeta(s)=\pi^{(s-1)/2}\Gamma(\frac{1-s}{2})\zeta(1-s)\]
With
\[ \xi(s):=\pi^{-s/2}\Gamma(\frac{s}{2})\zeta(s),\]
the functional equation may be rewritten as 
\begin{equation} \label{xi}
\xi(1-s)=\xi(s).
\end{equation}
The Riemann zeta-function has zeros at even negative integers $\{-2,-4, \dots \}$. These are referred to as \textit{trivial zeros}. \textit{The classic Riemann Hypothesis} states that \textit{ the nontrivial zeros of $\zeta(s)$ lie on the critical line} $\text{Re}(s)=1/2$.

\subsection{Zeta functions of curves over finite fields.} Let $\X=\X_g$ be a smooth, projective curve of genus $g$ over $\F_q$. Consider the generating function of the numbers of points $N_k:=\#\X(\F_{q^k})$:
\[ G_{\X}(T):=\sum_{k=1}^{\infty}\frac{N_k}{k}T^k.\]
The zeta function of $\X$ is defined as 
\[ \zeta_{\X}(T):=\text{exp}\left(G_{\X}(T)\right).\]
For example, let $\X=\P^1$. Then $N_k=q^k+1$, therefore for $|qT|<1$ and $|T|<1$ we get
\[G_{\P^1}(T)=\sum_{k=1}^{\infty}\frac{q^k+1}{k}T^k=\sum_{k=1}^{\infty}\frac{q^k}{k}T^k+\sum_{k=1}^{\infty}\frac{1}{k}T^k=-\log(1-qT)-\log(1-T).\] 
It follows that 
\[\zeta_{\P^1}(T)=\text{exp}\left(G_{\P^1}(T)\right)=\frac{1}{(1-T)(1-qT).}\]
It is known that the zeta function of $\X$ may be written as 
\[\zeta_{\X}(T)=\frac{L_{\X}(T)}{(1-T)(1-qT)}\]
where $L_{\X}(T)$ is a monic polynomial of degree $2g$ with integer coefficients. $L_{\X}(T)$ is called \textit{the L-polynomial of $\X$}. It satisfies the functional equation 
\[L_{\X}(T)=q^qT^{2g}L_{\X}(1/qT).\]
It follows that $L_{\X}(T/\sqrt q)$ is a degree $2g$, self-reciprocal polynomial, i.e. its coefficients satisfy $a_i=a_{2g-i}$ for $i=1,2,\ldots,g$. Let 
\[\xi_{\X}(s):=q^{sg}L_{\X}(q^{-s}).\]
The functional equation of $L_{\X}(T)$ yields 
\[\xi_{\X}(s)=\xi_{\X}(1-s).\]
\textit{The Riemann Hypothesis for finite fields}, proven by Weil in the 1940's, states that the roots of $L_{\X}(T)$ lie on the circle $|T|=1/\sqrt q$. Alternatively, the zeros of $\xi_{\X}(s)$ lie on the critical line $\text{Re}(s)=1/2$. The L-polynomial of $\X$ has a factorization
\[L_{\X}(T)=\prod_{i=1}^{2g}(1-\alpha_iT)\]
where $|\alpha_i|=\sqrt q$ for $i=1,2,\ldots,2g$.

\subsection{Zeta Function for Linear Codes} 
Motivated by analogies with local class field theory, Duursma introduced the zeta function of a linear code over a finite field. For $d\leq n$, denote the weight enumerator of an MDS code $C$ of length $n$ and minimum distance $d$ by $M_{n,d}(x,y)$. The dual $C^{\perp}$ is also an MDS code of length $n$ and minimum distance $d^{\perp}=n+2-d$.  Therefore, for $d\geq 2$, the weight enumerator of $C^{\perp}$ is $M_{n,n+2-d}(x,y)$. Let $M_{n,n+1}=x^n$. The MDS code with weight enumerator $M_{n,1}$ has dimension $n-d+1=n-1+1=n$, hence $C=\F^n_q$. It is easy to see that $M_{n,n+1}$ is the MacWilliams transform of $M_{n,1}$. We may think of $M_{n,1}$ as the weight enumerator of the zero code. The following proposition follows easily.
\begin{prop} The set $\{M_{n,1}, M_{n,2},\ldots,M_{n,n-1}, M_{n,n+1}\}$ is a basis for the vector space of homogeneous polynomials of degree n in $x,y$. Furthermore, this set is closed under MacWilliams transformations.
\end{prop}
If $C$ is an $[n,k,d]_q$-code, then one can easily see that
\[A_C(x,y)=\sum_{i=d}^{n+1}a_{i-d}M_{n,i}=a_0M_{n,d}+\ldots+a_{n+1-d}M_{n,n+1}.\]
\begin{definition} \label{M_n basis} The zeta polynomial of $C$ is defined as $P(T):=a_0+a_1T+ \dots +a_{n-d+1}T^{n+1-d}$. The quotient 
\[ Z(t)=\frac{P(T)}{(1-T)(1-qT)}\]
is called \textit{the zeta function} of the linear code $C$
\end{definition}
The zeta polynomial $P(T)$ of an $[n,k,d]_q$-code $C$ determines uniquely the weight enumerator of $C$. The degree of $P(T)$ is at most $n-d+1$; the following theorem establishes the precise value of the degree.
\begin{thm}[Duursma \cite{D6}]
Let $[n,k,d]$ and $[n,k^{\perp},d^{\perp}]$  be the parameters of dual codes $C$ and $C^{\perp}$. Denote by $P(T), Z(T), P^{\perp}(T),Z^{\perp}(T)$ their zeta polynomials and zeta functions. Let $g=\gamma(C)=n-k-d+1,~g^{\perp}=\gamma (C^{\perp})=n-k^{\perp}-d^{\perp}+1$. Then 
 
 (a) $\text{deg}~P(T)=\text{deg}~P^{\perp}(T)=g+g^{\perp}=n+2-d-d^{\perp},$
 
 (b) $P^{\perp}(T)=P(1/qT)q^gT^{g+g^{\perp}},$
 
 (c) $Z^{\perp}(T)=Z(1/qT)q^{g-1}T^{{g+g^{\perp}-2}},$
 
 (d) $P(1)=1$.
 
 (e) The zeta polynomial of any MDS code is $P(T)=1$.
 \end{thm}
\begin{proof}  
Assume that $P(T)$ is of degree $r$, hence 
\[ A_C(x,y)=a_0M_{n,d}(x,y)+a_1M_{n,d+1}(x,y)+ \dots +a_rM_{n,d+r}(x,y).\]
Recall that if the weight enumerator of an MDS code is $M_{n,i}(x,y)$, the weight enumerator of its dual is $M_{n,n+2-i}(x,y)$. Formulated in light of MacWilliams Identity, 
\[ M_{n,n+2-i}(x,y)=q^{i-n-1}M_{n,i}(x-(q-1)y,x-y).\]
This leads to
\[ A_{C^{\perp}}=q^{-k}A_C(x+(q-1)y,x-y)=a_rq^{g-r}M_{n,n+2-d-r}+ \dots +a_0q^gM_{n,n+2-d}.\]
The minimum distance of $C^{\perp}$ is $d^{\perp}$, hence the basis expansion of $A_{C^{\perp}}$ starts with $M_{n,d^{\perp}}$. It follows that $n+2-d-r=d^{\perp}$ or $r=n+2-d-d^{\perp}$, therefore 
\[ P^{\perp}(T)=a_rq^{g-r}+ \dots +a_0q^gT^r.\]
So, both $P(T)$ and $P^{\perp}(T)$ are of degree $r=g+g^{\perp}=n+2-d-d^{\perp}$. Now 
\[ 
\begin{split}
P^{\perp}(T) & =a_rq^{g-r}+ \dots +a_0q^gT^r \\
& =a_rq^{-g^{\perp}}+a_{r-1}q^{-g^{\perp}+1}T+ \dots +a_gT^{g^{\perp}}+ \dots +a_0q^gT^r \\
& =q^gT^r(a_0+ \dots +a_gq^{-g}T^{-g}+ \dots +a_rq^{-r}T^{-r}) \\
& =q^gT^rP(1/qT)=q^gT^{g+g^{\perp}}P(1/qT).
\end{split}
\]
The equation for zeta functions follows easily. By comparing the coefficients of $x^n$ on both sides we obtain $\displaystyle{\sum_{i=d}^{n+1}a_{i-d}=1}$, i.e. $P(1)=1$. Part (e) of the theorem is clear.
\end{proof}
\begin{cor}
The zeta function of an MDS code 
\[ \frac{1}{(1-T)(1-qT)}=\sum_{j=0}^{\infty}\frac{q^{j+1}-1}{q-1}T^j\]
is the rational zeta function over $\F_q$.
\end{cor}
\begin{cor} \label{zeta dual}
The zeta polynomial and the zeta function of a self-dual code $C$ satisfies the following functional equation 
\[P(T)=q^gT^{2g}P(1/qT),~Z(T)=q^{g-1}T^{2g-2}Z(1/qT).\]
\end{cor}

\noindent Notice that the zeta polynomial of a linear code and the L-polynomial of a genus $g$ curve over $\F_q$ satisfy the same functional equation. 

Here is another characterization of the zeta polynomial of a linear code.
\begin{prop}
Let $C$ be a linear code $C$ of length $n$ and minimum distance $d$. Assume that the minimum distance $d^{\perp}$ of its dual $C^{\perp}$ satisfies $d^{\perp}\geq 2$. Then, the zeta polynomial of $C$ is the only polynomial $P(T)$ of degree $n+2-d-d^{\perp}$ such that the generating function 
\[ \frac { [y(1-T)+xT]^n}{(1-T)(1-qT)}P(T)\] has $T$-expansion  
\[ \ldots +      \frac{ A_C(x,y) - x^n}   {q-1}T^{n-d} +\ldots.\]
\end{prop} 

\begin{proof} 
The proof presented here is due to Chinen \cite{Ch1}. Define $c_k(x,y)$ by 
\begin{equation} \label{c_k}
\frac{(y(1-T)+xT)^n}{(1-T)(1-qT)}=\sum_{k=0}^\infty c_k(x,y)T^k
\end{equation}
First note that
\[\frac{1}{(1-T)(1-qT)}=\sum_{j=0}^{\infty}{\frac{q^{j+1}-1}{q-1}T^j}\]
\noindent and
\[ (y(1-T)+xT)^n=\sum_{i=0}^{n}{ \begin{binom} {n}{i} \end{binom}y^{n-i}(x-y)^iT^i}.\]
\noindent Therefore,
\begin{equation} \label{explicitc_k}
c_k(x,y)=\sum_{i+j=k}\frac{q^{j+1}-1}{q-1}{n \choose i}y^{n-i}(x-y)^i.
\end{equation}

Note that  $n+2-d-d^{\perp}\leq n-d$ since $d^{\perp}\geq 2$. A polynomial $P(T)=\sum_{i=0}^{n-d}a_iT^i$ satisfies the identity 
\[ \frac{[y(1-T)+xT]^n} {(1-T)(1-qT) } \sum_{i=0}^{n-d}a_iT^i= \sum_{k=0}^{\infty}c_kT^k\sum_{i=0}^{n-d}a_iT^i=\ldots +  \frac{ A_C(x,y) - x^n}{q-1}T^{n-d} +\ldots.\]
 if and only if  
\begin{equation} \label{c_k basis}
\sum_{i=0}^{n-d}a_ic_{n-d-i}(x,y)=\frac{1}{q-1}\sum_{i=d}^n A_ix^{n-i}y^i.
\end{equation}

\noindent Expansion of $c_0(x,y), c_1(x,y), \dots ,c_{n-d}(x,y)$ as homogeneous polynomials of $x,y$ yields: 
\begin{equation} \label{c_{n-d}}
\begin{split}
c_0(x,y) & =b_{0,0}y^n, \\
c_1(x,y) & =b_{1,1}xy^{n-1}+b_{1,0}y^n, \\
 \dots  & \dots  \dots  \dots  \dots \\
c_{n-d}(x,y) &=b_{n-d,n-d}x^{n-d}y^d+b_{n-d,n-d-1}x^{n-d-1}y^{d+1}+ \dots +b_{n-d,0}y^n. 
\end{split}
\end{equation} 

\noindent The coefficients are obtained by comparison with Eq.~\eqref{explicitc_k}:

\begin{equation} \label{b_{k,l}}
b_{k,l}=\sum_{i=l}^k \frac{q^{k-i+1}-1}{q-1}(-1)^{i-l}{n \choose i}{i \choose l}~\text{for}~0\leq l\leq k\leq n-d,
\end{equation}

\noindent and $b_{k,l}=0$ otherwise. Consider the following matrices: 
\[ B:=(b_{k,l})^{t},  \, \,       ~{\bf a} :=(a_{n-d},a_{n-d-1},\ldots,a_0)^t,\]
and 
\[ {\bf A}:=\displaystyle{\frac{1}{q-1}}(A_n,A_{n-1},\ldots,A_d)^t.\]

\noindent We write the equations Eq.~\eqref{c_{n-d}} in the form
\begin{equation}
c_k(x,y)=\sum_{l=0}^k b_{k,l}x^ly^{n-l},~k=0,1, \dots ,n-d,
\end{equation}
\noindent and substitute them in Eq.~\eqref{c_k basis}. Comparing the coefficients of monomials on both sides of the equation shows that Eq.~\eqref{c_k basis} is equivalent to the system of $n-d+1$ linear equations in the $n-d+1$ variables $a_0,a_1, \dots a_{n-d}$: 
\[ B\bf a=\bf A.\]
The diagonal entries of $B$ are binomial coefficients $\displaystyle{b_{i,i}={n \choose i}}$, which are nonzero. It follows that $B$ is nonsingular. Therefore $\bf a$, hence $P(T)$ exist and is unique.
\end{proof}

\begin{cor}\label{AD} 
If the minimum distance $d^{\perp}$ of the dual code $C^{\perp}$ satisfies $d^{\perp}\geq 2$, then 
\[ \displaystyle{P(0)=(q-1)^{-1}{n \choose d}^{-1}A_d},~ \text{and}~\frac{A_{d+1}}{q-1}={n \choose {d+1}} ( P(0)(q-d)+P'(0)).\]
\end{cor}

\begin{proof} 
See also Corollary 97 in \cite{JK}. The proof follows easily from the above linear system $B\bf a=\bf A.$
\end{proof}
\begin{rem}
Let $C$ be a linear code such that $d^\perp=1$. Let $e_1,e_2,\ldots,e_r$ be all codeword in $C^{\perp}$ of weight one. For $j=1,2,\ldots,r$, let $i_j$ be the only position where $e_j$ has a nonzero coordinate.  Then, every codeword of $C$ has $0$ in the these $i_j$-th positions. We say that the code $C$ is \textit{degenerate}.  If we puncture/delete the coordinates in positions $i_j,~j=1,2,\ldots,r$, we get a new code $C'$ of length $n-r$ and weight distribution $(1,0,0,\ldots,A_d,\ldots,A_{n-r})$. Note that $x^rA_{C'}(x,y)=A_C(x,y)$. This new code is non degenerate, hence $d^\perp\geq 2$. The last theorem may be used as a definition of zeta polynomials for non-degenerate codes.
\end{rem}
\begin{definition}
A degree $m$ polynomial $f(x)=a_0+a_1x+ \dots +a_mx^m$ is called \textit{self-reciprocal} if 
\[ f(x)=x^m f(1/x),\] 
i.e. if and only if  the following equality of $(m+1)$-tuples holds 
\[ (a_0,a_1, \dots ,a_m)=(a_m, \dots ,a_1,a_0).\]
\end{definition}

\noindent Formally self orthogonal codes lead to self-reciprocal polynomials.

\begin{prop} \label{self reciprocal} 
If $P(T)$ is the zeta polynomial of a formally self-orthogonal code, then $P(T/{\sqrt q})$ is a self-reciprocal polynomial.
\end{prop}

\begin{proof}
Let $C$ be a formally self-orthogonal. Recall that this means 
\[ A_C(x,y)=A_{C^\perp}(x,y).\]
It follows that $C$ and $C^\perp$ have the same zeta polynomial $P(T)$. It has degree $2g$ and satisfies 
\[ P(T)=q^gT^{2g}P(1/qT).\]
Define the degree $2g$ polynomial $P^s(T):=P(T/{\sqrt{q}}).$ Then, 
\[ P^s(T)=q^g(T/{\sqrt q})^{2g}P(1/T{\sqrt q})=T^{2g}P^s(1/T).\]
\end{proof}

\subsection{Riemann zeta function versus zeta function  for self-dual codes}

We saw in Corollary \ref{zeta dual} that for a self-dual code $C$,  
\[ Z(T)=q^{g-1}T^{2g-2}Z(1/qT),\] which for 
\[ z(T):=T^{1-g}Z(T), \]
may be written as 
\[ z(T)=z(1/qT).\]
Now let 
\[ \zeta_C(s):=Z(q^{-s}),~\text{and}~\xi_C(s):=z(q^{-s}).\]
We obtain 
\[ \xi_C(s)=\xi_C(1-s),\]
which is the same symmetry equation as Eq.~\eqref{xi}. We note that $\zeta(s)$ and $\xi(s)$ have the same zeros.

The zeroes of the zeta function of a linear code $C$ are useful in understanding possible values of its minimum distance $d$.

\begin{prop}
Let $C$ be a linear code with weight distribution vector $(A_0,A_1,\ldots,A_n).$ Let $\alpha_1,\ldots, \alpha_r$ be the zeros of the zeta polynomial $P(T)$ of $C$ Then
\[ d=q-\sum_{i}{\alpha_i^{-1}}-\frac{A_{d+1}}{A_d}\frac{d+1}{n-d}. \]
In particular,
\[ d\leq q- \sum_ {i}{\alpha_i^{-1}}. \]
\end{prop}

\begin{proof}
The first statement follows from Corollary \ref{AD}. The second statement is an easy consequence of the first. 
\end{proof}

\begin{definition}
A self-dual code $C$ is said to satisfy Riemann hypothesis if the real part of any zero of $\zeta_C(s)$ is $1/2$, or equivalently, the zeros of the zeta polynomial $P_C(T)$ lie on the circle $|T|=1/\sqrt q$, or equivalently, the roots of the self-reciprocal polynomial (see Proposition \ref{self reciprocal} above) $P_C(T/\sqrt q)$ lie on the unit circle.
\end{definition}

\begin{exmp}
Consider the binary code generated by the following matrix:
\begin{equation}
\left(\begin{array}{cccccccc}
1 & 0 & 0 & 0 & 0 & 1 & 1 & 1  \\
0 & 1 & 0 & 0 & 1 & 0 & 1 & 1  \\
0 & 0 & 1 & 0 & 1 & 1 & 0 & 1 \\
0 & 0 & 0 & 1 & 1 & 1 & 1 & 0  \\
\end{array} \right).
\end{equation}
The code above is in fact the $[8,4,4]$ extended Hamming code with
$|C|=16$ codewords. Then the Zeta function of this code is
\[Z(T)=\frac{2T^2+2T+1}{5(1-2T)(1-T)}.\] 
The roots of $2T^2+2T+1=0$ are $(-1/2)\pm(1/2)i$, and they both lie on the circle $|T|=1/\sqrt 2$.
\end{exmp}

While Riemann hypothesis is satisfied for curves over finite fields, in general it does not hold for linear codes. A result that generates many counterexamples may be found in \cite{JK}. There is a family of self-dual codes that satisfy the Riemann hypothesis which we are about to discuss. The theory involved in this description holds in more generality than linear codes and their weight enumerators, Namely, it applies to the so called \textit{virtual weight enumerators}. 

\subsection{Virtual Weight Enumerators}  

There is a straightforward generalization of the weight enumerator $A_C(x,y)$ of a linear code $C$.

\begin{definition}  
A homogeneous polynomial
\[ F(x,y)=x^n+\sum_{i=1}^{n}{f_ix^{n-i}y^i}\]
with complex coefficients is called a \textit{virtual weight enumerator}. The set 
\[ \{0\}\cup \{i:f_i\neq 0\}\]
 is called its \textit{support}. If 
 \[ F(x,y)=x^n+\sum_{i=d}^n f_ix^{n-i}y^i, \]
with $f_d\neq 0,$ then $n$ is called the \textit{length} and $d$ is called the \textit{minimum distance} of $F(x,y)$. 
\end{definition}

Let $C$ be a self-dual linear $[n,k,d]$-code. Recall that $n$ is even, $k=n/2$ and its weight enumerator satisfies MacWilliams' Identity.  A virtual generalization of $A_C(x,y)$ is straightforward. A virtual weight enumerator $F(x,y)$ of even degree that is a solution to MacWilliams' Identity
\[F(x,y)=F \left(\frac{x+(q-1)y}{\sqrt{q}},\frac{x-y}{\sqrt{q}} \right),\]
is called \textit{virtually self dual} over $\F_q$ with \textit{genus} $\gamma(F)=n/2+1-d$. Although a virtual weight enumerator in general does not depend on a prime power $q$, a virtually self-dual weight enumerator does. 

\begin{prob}
Find the conditions under which a (self-dual) virtual weight enumerator with positive integer coefficients arises from a (self-dual) linear code. 
\end{prob}

The zeta polynomial and the zeta function of a virtual weight enumerator are defined as in the case of codes.

\begin{prop} [\cite{Ch}]
Let $F(x,y)$ be a virtual weight enumerator of length $n$ and minimum distance $d$. Then, there exists a unique function $P_F(T)$ of degree at most $n-d$ which satisfies the following
\[\frac{(y(1-T)+xT)^n}{(1-T)(1-qT)}P_F(T)=\ldots +\frac{F(x,y)-x^n}{q-1}T^{n-d} +\ldots\]
\end{prop}

\noindent The polynomial $P_F(T)$ and the function 
\[ Z_F(T):=\frac{P(T)}{(1-T)(1-qT)},\]
 are called respectively \textit{the zeta polynomial and the zeta function of the virtual weight enumerator} $F(x,y)$.

\begin{definition}
A virtual self-dual weight enumerator satisfies the Riemann hypothesis if the zeroes of its zeta polynomial $P_F(T)$ lie on the circle $|T|=1/\sqrt{T}$.
\end{definition}

There is a family of virtual self-dual weight enumerators that satisfy Riemann hypothesis. It consists of enumerators that have certain divisibility properties.

\begin{definition}
Let $b>1$ be an integer. If supp$(F)\subset b \Z$, then $F$ is called $b$-\textit{divisible}.
\end{definition}

\begin{thm}[Gleason-Pierce]
Let 
\[ \displaystyle{F(x,y)=x^n+\sum_{i=d}^n f_ix^{n-i}y^i}\]
 be a $b$-divisible, virtually self-dual weight enumerator over $\F_q$. Then either

\begin{description}
\item [I] $q=b=2$ or
\item [II] $q=2, b=4$ or
\item [III] $q=b=3$ or
\item [IV] $q=4, b=2$ or
\item [V] $q$ is arbitrary, $b=2$, and $F(x,y)=(x^2+(q-1)y^2)^{n/2}$.

\end{description}
\end{thm}

\begin{proof}
We follow \cite{amt}. Let $\epsilon$ be a primitive $b$-th root of unity. Then,
\[ F(x,\epsilon y)=F(x,y).\]
Let $G\lhd PGL(2,\C)$ the subgroup generated by the following matrices 

\[ 
E:=\left( \begin{array}{cc}
1 & 0  \\
0 & \epsilon 
\end{array} \right),~
M=\left( \begin{array}{cc}
1 & q-1  \\
1 & -1 
\end{array} \right).
\]

\noindent The linear action of $G$ on the projective space $\P^1(\C)$ descends into an action on the zero locus 
\[ Z(F):=\{(x,y)\in \P^1(\C)~:~F(x,y)=0\}.\]
We notice that $(1,0)\notin Z(F)$. There are no fixed points for the action of $G$, therefore 
\[ \#(Z(F))>1.\]
Recall that a linear action on $\P^1(\C)$ is determined by the image of three points. It follows that if $\#(Z(F))\geq 3$ then $G$ is finite.

\textbf{Case 1}: $b=2$ and $F(x,y)$ has only two roots. Notice that when $b=2$ both $(0,1),~(1,0)\notin Z(F)$. Let $(\alpha,1)$ and $(-\alpha,1)$ be the roots of $F(x,y)$, $\alpha\neq 0$. Since $\epsilon=-1$, the matrix $E$ permutes these two roots. On the other hand 
\[ M\cdot (\alpha,1)=(\frac{\alpha+q-1}{\alpha-1},1)\]
must be either $(\alpha,1)$ or $(-\alpha,1)$. 

If 
\[ \frac{\alpha+q-1}{\alpha-1}=\alpha\]
then one can easily see that $(\alpha,1),~(-\alpha,1)$ and $M\cdot (-\alpha,1)$ are three roots of $F(x,y)$, violating the assumption that $F$ has two roots. Hence $M\cdot (\alpha,1)=(-\alpha,1)$, i.e. 
\[ \frac{\alpha+q-1}{\alpha-1}=-\alpha.\]
Hence $\alpha=\pm i \sqrt{q-1}$, therefore $(i \sqrt{q-1},1),~(-i \sqrt{q-1},1)$ are the only roots of $F(x,y)$. Since $b=2$,  $F(x,y)$ is a polynomial of $x^2,y^2$. It follows that 
\[ F(x,y)=[(x+i \sqrt{q-1}y)(x-i \sqrt{q-1}y)]^{n/2}=[(x^2+(q-1)y^2]^{n/2}.\]
This is case {\bf V} in the theorem.

\textbf{Case 2}: $b=2$ and $F(x,y)$ has more than two roots, or $b\geq 3$. Notice that if $b\geq 3$ then $\#(Z(F))\geq 3.$ Indeed, if $\alpha\neq 0$ and $(\alpha,1)\in Z(F)$, then $(\alpha \epsilon^i,1)\in Z(F)$ for $i=0,1, \dots ,b-1$. If $(0,1)\in Z(F)$, then $M\cdot (0,1)=(1-q,1)\in Z(F)$ hence $((1-q)\epsilon^i,1)\in Z(F)$ for $i=0,1, \dots ,b-1$. 

It follows that in this case $G$ is finite. Therefore every element of $G$ has finite order. Let $k$ be the order of the matrix 
\[ 
ME=\left( \begin{array}{cc}
1 & \epsilon(q-1)  \\
1 & -\epsilon 
\end{array} \right).
\]
The eigenvalues $\lambda_1, \lambda_2$ of $ME$ satisfy the characteristic equation 
\[ \lambda^2+(\epsilon-1)\lambda-\epsilon q=0.\]
 The equation $(ME)^k=cI$ implies that $({\lambda_1}/{\lambda_2})^k=1.$ Hence $\epsilon, \lambda_1/\lambda_2, \lambda_2/\lambda_1$ are algebraic integers, therefore 
 \[ \left(2+\frac{\lambda_1}{\lambda_2}+\frac{\lambda_2}{\lambda_1}\right)=\frac{(\lambda_1+\lambda_2)^2}{\lambda_1\lambda_2}=-\frac{(\epsilon-1)^2}{\epsilon q}\]
is also an algebraic integer. It follows that 
\[ -\frac{(\epsilon-1)^2}{\epsilon q},~\frac{(\epsilon-1)^2}{q}\in \Z[\epsilon].\]
If $b$ is not a prime power then $\epsilon-1$ is a unit in $\Z[\epsilon]$, therefore $\displaystyle{\frac{(\epsilon-1)^2}{q}\notin \Z[\epsilon].}$ If $b$ is a power of a prime number $p$, then in $\Z[\epsilon]$ we have an equality of ideals 
\[ (1-\epsilon)^{\phi(b)}=(p),\]
 where $\phi$ denotes the Euler function. It follows that $\phi(b)=1$ or $\phi(b)=2$. If $\phi(b)=1$ then $b=2$. Therefore $\epsilon=-1$, hence $4/q$ is an integer. In this case $q=2$ or $q=4$. If $\phi(b)=2$ then $b=3,4,6$. But $b\neq 6$, otherwise $-(\epsilon-1)/q\epsilon=1/q$, would be an algebraic integer! If $b=3$, then $\displaystyle{\frac{-(\epsilon-1)^2}{\epsilon}}=3$. Therefore $q=3$. If $b=4$ then $\epsilon=i$ and $2/q$ must be an integer. It follows that $q=2$.
\end{proof}

\begin{definition} 
A $b$-divisible virtually self-dual weight enumerator $F(x,y)$ over $\F_q$ is called
\begin{description}
\item [Type I] if $q=b=2,~2|n$.
\item [Type II] if $q=2, b=4, 8|n$.  
\item [Type III] if $q=b=3, 4|n$.
\item [Type IV] if $q=4, b=2, 2|n$.
\end{description}
\end{definition}

\begin{thm}[Mallows-Sloane-Duursma]  \label{bound}
If $F(x,y)$ is a $b$-divisible self-dual virtual enumerator with length $n$ and minimum distance $d$, then
\[
d \leq 
\left\{ 
\begin{split}
& 2 \left[  \frac{n}{8}\right] +2,  \quad & \text{ if  F is Type I},\\
& 4 \left[ \frac{n}{24}\right]  +4, \quad & \text{ if F is Type II}, \\
& 3  \left[ \frac{n}{12}\right] +3, \quad &  \text{ if F is Type III}, \\
& 2  \left[ \frac{n}{6}\right] +2, \quad &  \text{ if F is Type IV}.\\
\end{split}
\right.
\]
\end{thm}

\noindent See  \cite{JK} for details of the proof.

\begin{definition}
A virtually self-dual weight enumerator $F(x,y)$ is called \textit{extremal} if the bound in Theorem \ref{bound} holds with equality. 
\end{definition}

\begin{definition}
A linear code $C$ is called $b$-divisible, extremal, Type I, II, II, IV if and only if  its weight enumerator has the corresponding property.
\end{definition}

\noindent The zeta functions of all extremal virtually self-dual weight enumerators are known; see \cite{D3}. 
The following result can be found in \cite{D3}. 

\begin{prop} 
All extremal type IV virtual weight enumerators satisfy the Riemann hypothesis.
\end{prop}

\noindent For all other extremal enumerators, Duursma has suggested the following conjecture in \cite{D4}.

\begin{prob} 
Prove that any extremal virtual self-dual weight enumerators of type I-III satisfies the Riemann hypothesis.
\end{prob}

\subsection{Formal weight enumerators}  

Fomal weight enumerators are introduced by Chinen in \cite{Ch1}. They are similar to virtual weight enumerators of type II. In this section we discuss the zeta polynomials and its fundtional equation, as well as Riemann hypothesis for extremal formal weight enumerators. All the definitions and the results may be found in \cite{Ch1}\cite{Ch}.
\begin{definition}(Chinen (\cite{Ch1})
A homogeneous polynomial $\displaystyle{W(x,y)=\sum_{i=1}^nW_ix^{n-i}y^i}$ is called a \textbf{formal weight enumerator} if the following two conditions are satisfied: 

\begin{description}

\item [(a)] If $W_i\neq 0$ then $4|i$, and 

\item [(b)] $\displaystyle{W\left(\frac{x+y}{\sqrt 2},\frac{x-y}{\sqrt 2}\right)=-W(x,y)}$.

\end{description}
\end{definition}

Let $\C [x, y]$ be the polynomial ring in two variables and $PGL_2 (\C)$ acting on $\C [x, y]$ by a linear change of coordinates, i.e.,  
for a matrix $M=\begin{pmatrix} a & b \\ c & d \end{pmatrix}$, we have 
\[ f^M (x, y) = f(ax+by, cx+dy). \]
Let $G_8$ be the subgroup of $PGL_2 (\C)$ generated as follows:
\[ G_8 = \left\langle \sigma_1=\frac {1-i} 2 \begin{pmatrix} 1 & -1 \\ 1 & 1 \end{pmatrix}, \quad \sigma_2=\begin{pmatrix} -i & 0 \\ 0 & 1 \end{pmatrix} \right\rangle \] 
with $i^2=-1$. Weight enumerators of type II curves and formal weight enumerators lie in the invariant polynomial ring $\C[x,y]^{G_8}$. For the later ones, the invariance under the action of 
\[\sigma_2(\sigma_1^2\sigma_2^3)^2=\begin{pmatrix} 1 & 0 \\ 0 & i \end{pmatrix}\]
explains condition (a).

\begin{lem} The following statements hold true:

i)  The invariant ring $\C [x, y]^{G_8}$ is generated by the polynomials 
\[ W_8(x,y)= x^8 + 14x^4y^4 + y^8, \textit{  and } W_{12}(x,y)= x^{12} - 33x^8y^4 - 33x^4y^8 + y^{12}\]

ii) A formal weight enumerator is a symmetric polynomial, i.e., $W(x, y) = W(y, x)$. 

iii) A formal weight enumerator $W(x, y)$ can be written as $W(x, y) = g (\bar x, \bar y)$, where $\bar x = x^4$ and $\bar y = y^4$ and $g \in \C [\bar x, \bar y]$.
\end{lem}

\proof
It is easy to check that $W_8$ and $W_{12}$ are fixed by the generators of $G_8$. To show that $\C [x, y]^{G_8} = \C [W_8, W_{12}]$ we have to show that the extension 
$\C [x, y]/ \C [W_8, W_{12}]$ has degree $| G_8 |$.  We leave this as an exercise. 

Part ii) is an immediate consequence of Part i), since any formal weight enumerator is generated by $W_8$ and $W_{12}$ which are both symmetric in $x$ and $y$. The same can be said for Part iii). 
\qed

Notice that $W_8$ is the weight enumerator of the extended Hamming code, which is a type II code. The generator $W_{12}$ is formal weight enumerator. In general, a formal weight enumerator is ${W_8}^sW_{12}^{2t+1}$ for positive integers $s,t$, and linear combinations of such. It follows that a formal weight enumerator has degree $4(\text{mod~}8)$ and consists of an even number of terms.

If $W(x,y)=x^n+\sum_{i=d}^nW_ix^{n-i}y^i$ with $W_d\neq 0$, then $n$ is called \textit{the length} and $d$ \textit{the minimum distance} of $W(x,y)$. Set $q=2$ and define 
\[W^{\perp}=\displaystyle{W\left(\frac{x+y}{\sqrt 2},\frac{x-y}{\sqrt 2}\right)}\]
Just as with virtual weight enumerators, there exists a zeta polynomial $P^{\perp}(T)$ for $W^{\perp}(T)$ which satisfies 
\[P^{\perp}(T)=P(1/2T)2^gT^{2g},\]
where $g=n/2+1-d$. From the definition, $P^{\perp}(T)$ must coincide with the zeta polynomial of $-W(x,y)$. We obtain the following 

\begin{prop}
The zeta polynomial of a formal weight enumerator $W(x,y)$ satisfies
\[P(T)=-P(1/2T)2^gT^{2g}\]
\end{prop}

Recall that weight enumerators of type II curves also lie in $\C [x, y]^{G_8}$. In contrast to formal weight enumerator, the zeta polynomial of a type II curve satisfies 
\[P(T)=P(1/2T)2^gT^{2g}\]
The last proposition can be used to find the roots of the zeta polynomial for a formal weight enumerator. They are $\alpha_1,1/2\alpha_1,\ldots \alpha_s,1/2\alpha_s$ for some $s$ and $\alpha_j\neq \pm1/{\sqrt 2}$, as well as $\pm1/{\sqrt 2}$ which occur in odd multiplicity. The proof is similar to \cite[Thm V.1.15]{st}.

\begin{thm}
For any formal weight enumerator of length $n$ and minimum distance $d$, we have
\[d\leq 4\left[\frac{n-12}{24}\right]+4.\]
\end{thm}

A formal weight enumerator is called \textit{extremal} if the above holds with equality.

\begin{prob}
Prove that any extremal formal weight enumerator satisfies the Rieman hypothesis, i.e. all roots of the zeta polynomial have absolute value $1/{\sqrt 2}$.
\end{prob}

\section{Algebraic Geometry Codes and their weight distributions}

\subsection{Divisors on algebraic curves}

Let $\X$ be an algebraic curve defined over $\F_q$, $\F=\F_q (\X)$ its function field of rational functions, and $\mathcal P_{\F}$ the set of places of $\X$. An integral linear combination $G:=\sum_{i} m_iQ_i, Q_i\in {\mathcal P}_{\F}$ is called a \textit{divisor}. The set $\supp(G):=\{Q_i~|~m_i\neq 0\}$ is called  \textit{the support} of $G$. If all $a_i\geq 0$, we call $G$  \textit{effective} and write $G\geq 0$. The sum of all integer coefficients $\sum_{i}m_i$ of the divisor $G$ is called {\it the degree} of the divisor $G$. Denote by $D(\F)$ the abelian group of divisors and by $E(\F)$ the semigroup of effective divisors. The set of \textit{principal divisors} $(f)$ for $0\neq f\in \F$ forms a subgroup of $D(\F)$. The quotient $D(\F)/P(\F)$ is \textit{the divisor class group} $C(\F)$, it is finitely generated of the form $C(\F)=\Gamma\times \Z$. The finite torsion subgroup $\Gamma$ consists of degree zero divisor classes. Let $E$ be a degree one divisor. A divisor class $[G]$ may be represented as $([G]-\deg G \cdot [E],\deg G)$. For a divisor $G$, denote 
\[\codeL (G):=\{f\in \F^*~:~(f)+D\geq 0\}\cup \{0\}.\]
the vector space of rational functions with pole divisor bound by $G$. It is well known (Riemann) that 
\[i(G):=\deg G-\dim \codeL(G)+1\geq 0,\]
for any divisor $G$. The number $i(G)$ is called {\it the index of speciality} of $G$. The maximum of these indexes for all divisors is called {\it the genus} $g$ of the curve $\X$. 

\subsection{Algebraic Geometry Codes}

Let $P_1, \ldots, P_n$ be pairwise different rational places, and $D = P_1 + \cdots + P_n$. Let $G=\sum_{i}m_iQ_i$ be a divisor such that $\supp(G) \cap \supp(D) = \emptyset$.  

The following algebraic geometry codes have been introduced by Goppa in the eighties:

\begin{enumerate}

\item $C_\codeL(D,G):=\{ (f(P_1), ... , f(P_n))~|~ f \in \codeL(G) \} \subset \F_q^n$.

\item  $C_{\Omega}(D,G):=\{(\res{P_1}{\omega},...,\res{P_n}{\omega})~|~ \omega \in \Omega_F(G-D)\}\subset \F_q^n.$

\item If $P\notin \supp(D$) and $m$ is an integer, $C_\codeL(D, mP)$ is called {\it one point code of level $m$}.

\end{enumerate}

\noindent Consider the evaluation map 
\[\varphi: \codeL(G)  \to  \F_q^n,~f \mapsto  (f(P_1), ..., f(P_n)).\]
The function $f\in \codeL(G)$ has poles only on the support of the divisor $G$. But $\supp(G) \cap \supp(D) = \emptyset$, therefore $f(P_i), i=1,2, \dots ,n$ belong to some extension of $\F_q$.  This extension has degree one since $P_1, \ldots, P_n$ are rational places. It follows that $f(P_i)\in \F_q,~\varphi$ is a well-defined map and 
\[C_\codeL (D,G) = \varphi (\codeL (G)) .\] 
One can easily see that  
\[\ker~\varphi=\{ f \in \codeL(G):~ f(P_i)=0,~i=1,2, \dots ,n\}=\codeL(G-D),\]
therefore $$\text{dim}~\codeL(D,G)=\text{dim}~\codeL(G)-\text{dim}~\codeL(G-D).$$ It follows that $C_\codeL(D,G)$ is a linear $[n,k,d]$ code of dimension
\[ k  =  \dim ~\codeL(G) - \dim ~\codeL(G-D).\]

\begin{prop}
$C_\codeL(D,G)^{\perp} = C_{\Omega}(D,G)$ under the standard Euclidean pairing in $\F_q^n$.  
There exists a Weil differential $\eta$ such that $C_{\Omega}(D,G)=C_\codeL(D,D-G+(\eta))$.
\end{prop}

The reader can check the details of the proof at \cite[Prop II.2.10]{st}. Now we are ready to define various algebraic geometry codes. 

\begin{definition}
A code $C=[n, k, d]$ is called \textbf{weakly algebraic geometry code} (WAG)  of genus $g$ if it can be represented as $C_\codeL (D,G)$ for some curve $\X_g$ of genus $g$. 
If $\deg G <n$ then $C$ is called simply an \textbf{AG code}. The code $C$ is called \textbf{strongly algebraic geometry code} (SAG)   if $2g-2<\deg G < n$.
\end{definition}

It can be shown that very linear code is a WAG code, but not all linear codes are AG codes; see \cite{TVT}.   

\begin{prop}
If $C$ is an $[n,k,d]_q$-AG code then $k = \dim \codeL(G) \geq \deg G + 1 - g$ and $d\geq n-\deg G$, hence $n+1-g\leq k+d\leq n+1.$ If $C$ is a SAG code then $k=\deg G+1-g.$
\end{prop}

\begin{proof}
If $C=C_\codeL(D,G)$ is WAG, then $\deg(G-D)<0,$ hence $\dim \codeL(G-D)=0$. It follows that $ \varphi: \codeL(G) \rightarrow C_\codeL(D,G)$ is injective and $C_\codeL(D,G)$ has dimension 
\[k = \dim \codeL(G) \geq \deg G + 1 - g.\]
If $C$ is SAG, then $2g-2 < \deg G < n$ and by the Riemann-Roch theorem 
\[ k=\deg G+1-g.\]
Let $G=G_1-G_2$ with $G_1\geq 0,~G_2\geq 0$.  A non-zero function $f\in \codeL(G)$ has at most $\deg (G_1)$ poles and at least $\deg (G_2)$ zeros on $\supp (G)$. It follows that the number of its zeros outside $\supp~D$ is at most $\deg G_1-\deg G_2=\deg G$. Therefore $\phi(f)$ has at least $n-\deg G>0$ non-zero coordinates, and thus $d\geq n-\deg G.$ The double inequality follows from the Singleton bound. 
\end{proof}

The numbers $k_c=\deg G+1-g$ and $d_c=n-\deg G$ are called respectively \textit{the designed dimension} and \textit{the designed minimum distance} of the AG code $C_\codeL(D,G)$. Notice that $d_c+k_c=n-g+1$, therefore the genus of the curve measures how far the WAG code is from being an MDS code. If $g=0$ then $k+d=n+1$, hence every rational AG code is MDS.

\begin{prop}$C_{\Omega}(D,G)$ is an $[n, k', d']$ code with parameters
\[ k' = i(G) - i(G-D), \quad     d'  \ge  \deg G-(2g-2).\]
If in addition, $\deg G >2g-2$ then $k' \geq n+g-1-\deg G$. If, moreover, $2g-2<\deg G<n$ then $k'=n+g-1-\deg G$.
\end{prop}

Details of the proof can be found at \cite[Thm II.2.7]{st}.

\subsection{Weight distributions of AG codes}
Computing the weight distribution of a linear code is generally a difficult problem. However, the extra structure of AG codes allows he weight distribution problem to be reformulated as one of the distribution of effective divisors over divisor classes, and here the group structure of the divisor classes may be employed. We follow closely Duursma \cite{D6}.

\begin{prop} (\cite{D6})\label{fiber}
Let $C=C_\codeL(D,G)$ be an $[n,k,d]_q$-AG code with weight distribution $(A_0=1,A_d,A_{d+1},\ldots A_n)$. For $i\leq n$, let $a_i=A_{n-i}$ be the number of codewords with precisely $i$ zeros. Then 
\[a_i=(q-1)\#\{H:H\sim G, H\geq 0, \#(\supp(H)\cap \supp(D))=i\}\]
\end{prop}

\begin{proof}
The evaluation map 
\[\varphi: \codeL(G)  \to  C_\codeL(D,G),~f \mapsto  (f(P_1), ..., f(P_n)).\]
is bijective since $C=C_\codeL(D,G)$ is an AG code. As $G$ and $D$ have disjoint support, the codeword $(f(P_1), ..., f(P_n))$ has $i$ zeros iff the support of $H=(f)+G$ has $i$ places from $\supp~D$. The function $f\in \codeL(G)$, hence the codeword $(f(P_1), ..., f(P_n))$, is determined by the divisor $H=(f)+G$ up to a non-zero scalar.
\end{proof}

It follows that \textit{to determine the weight distribution of an AG code $C_\codeL(D,G)$, one must study the effective divisors in the class of $G$ that have a precise number of places from $\supp~D$}.

As in Section 4.1, the divisor class group is finitely generated of rank one, i.e. it is isomorphic to $\Gamma\times \Z$ via the choice of a degree one divisor $E$. Here $\Gamma$ is the finite torsion subgroup of degree zero divisor classes. The divisor $G$ may be identified with $([G]=G-\deg G\cdot E, \deg G\cdot E)\in \Gamma\times \Z$. Let
\[L(T):=\sum_{r\geq 0}\sum_{h\in \Gamma}\#((h+rE)\cap E(\F))X^hT^r\]
be the generating function for the number of effective divisors in the divisor class $h+rE$. It should be considered as an element of $\C[\Gamma][[T]]$, i.e. a power series of $T$ with coefficients in the complex group algebra $\C[\Gamma]$ of the torsion group $\Gamma$. The characteristic functions $\{X^h:h\in \Gamma\}$ form a basis of $\C[\Gamma]$ as a $\C$-vector space.  Another basis of $\C[\Gamma]$ may be obtained using $\hat{\Gamma}$, the characters of $\Gamma$. For $\chi\in \hat{\Gamma}$, define
\[e_{\chi}=\frac{1}{\#(\Gamma)}\sum_{h\in \Gamma}\chi(-h)X^h\]
It is straightforward to show that $X^he_\chi=\chi(h)e_\chi,$ and using basic character theory
\[X^h=\sum_{\chi\in \hat{\Gamma}}\chi(h)e_\chi.\]
It follows that $\{e_\chi:\chi\in \hat{\Gamma}\}$ is a basis of orthogonal idempotents for $\C[\Gamma]$. We get the coordinates of $L(T)$ in these two bases
\[L(T)=\sum_{g\in \Gamma}L(T,h)X^h=\sum_{\chi\in \hat{\Gamma}}L(T,\chi)e_\chi.\]
The coordinate $L(T,h)$ is clear from the definition of $L(T)$. We notice that 
\[\sum_{h\in \Gamma}L(T,h)=Z(T)\]
where $Z(T)$ is the zeta function of the function field $\F$. The other coordinate has the following form 
\[L(T,\chi)=\prod_{P\in \mathcal P(\F)}\frac{1}{1-\chi([P])T^{\deg P}}\in \C[[T]]\]
After the substitution $T=q^{-s}$,  $L(T,\chi)$ is a Dirichlet L-series for the function field $\F$. 

Let $P$ be a rational place. Write it in the form $[P]+E$ with $[P]\in \Gamma$. For a subset $\mathcal P$ of rational places, define
\[\Lambda_{\mathcal P}(T)=\prod_{P\in \mathcal P}(1+X^{[P]}T)\in \C[\Gamma][T]\]
with its coordinate functions
\[\Lambda_{\mathcal P}(T)=\sum_{h\in \Gamma}\Lambda_{\mathcal P}(T,h)X^h=\sum_{\chi\in \hat{\Gamma}}\Lambda_{\mathcal P}(T,\chi)e_\chi.\]
\begin{thm}
The distribution over divisor classes of effective divisors that contain precisely a given number of places from $\mathcal P$ is given by
\[A_{\mathcal P}(U,T)=L(T)\Lambda_{\mathcal P}(U-T)\in\C[\Gamma][U](T)\]
Its coordinate function $A_{\mathcal P}(U,T,h)$ is the generating function for the number of effective divisors in the divisor class $h+(i+j)E$ with precisely $i$ places of $\mathcal P$ in its support.
\end{thm}

\begin{proof}
The Euler product decomposition of the distribution $L(T)$ is 
\[L(T)=\prod_{P\in \mathcal P(\F)}\left(\frac{1}{1-X^{[P]}T^{\deg P}}\right)\]
The contribution of a rational place $P\in \mathcal P$ in $A_{\mathcal P}(U,T)$ is 
\[\frac{1+X^{[P]}(U-T)}{1-X^{[P]}T}=\frac{X^{[P]}U}{1-X^{[P]}T}=1+X^{[P]}U+X^{2[P]}UT+X^{3[P]}UT^2+\ldots\]
Hence the variable $U$ keeps track of the precise number of places $P$ that contribute to a term of $A_{\mathcal P}(U,T)$.
\end{proof}

\begin{cor}
The coordinate function $A_{\supp~D}(U,T,[G])$ determines the weight distribution of the AG code $C_\codeL(D,G)$.
\end{cor}

One computes the coordinate functions $A_{\supp~D}(U,T,\chi)$ and then applies an inverse Fourier transform to recover the functions $A_{\supp~D}(U,T,g)$. If the zeta function of the function field $\F$ is known, then estimates of the weight distribution of an AG code may be obtained via their average.

\begin{thm}
If the zeta function of the function field $\F$ is $Z(T)$ then the average weight distribution 
\[\frac{1}{\#(\Gamma)}\sum_{h}A_{\supp D}(U,T,h)=\frac{1}{\#(\Gamma)}Z(T)(1+U-T)^n.\]
\end{thm}

MacWilliams' identity for the dual of an AG code may also be intrinsically expressed via the generating function $A(U,T)$. Define an involution on $\C[\Gamma]$ via $\overline{X^h}=X^{-h}$. Let $W$ denote the canonical divisor class on $\X$. 

\begin{prop}
(MacWilliams Identity) The distribution $A(U,T)$ satisfies a functional equation
\[A(U,T)=\overline{A(1/(U-T)+1/qT,1/qT)}X^{[W+D]}(U-T)^n(qT^2)^{g-1}\]
\end{prop}

If $G=h+aE$ and $G'=h'+a'E$ are divisors with $G+G'=W+D$, then weight distributions of $C_\codeL(D,G)$ and $C_\codeL(D,G')$ are given by the coefficients $A_{a-j,j,h}$ and $A_{a'-j,j,h'}$ of $U^{a-j}T^j X^h$ and $U^{a'-j}T^j X^{h'}$ in $A(U,T)$. They are related via the above proposition as in Eq.~\eqref{eq-8}. This can be used for AG codes since the dual of  $C_\codeL(D,G)$ is of the form $C_\codeL(D,G'=W+D-G)$ and $G+G'=W+D$. We get

\begin{prop} \label{dual weights}
The weight distributions of the dual codes $C_\codeL(D,G)$ and $C_\codeL(D,G')$ are determined by the combined set of coefficients $A_{a-j,j,h}$ and $A_{a'-j,j,h'}$ for $j=0,1,\ldots g-1$. The remaining coefficients can be computed via the MacWilliams' identity. 
\end{prop}

 \subsubsection{Rational AG codes} Let $C=C_\codeL(D,G)$ be an $[n,k,d]_q$-AG code of genus $g=0$. Since there are $q+1$ rational places on a genus zero curve over $\F_q$, we get $n\leq q+1$. The dimension $k=0$ iff $\deg G<0$, and $k=n$ iff $\deg G>n-2$. For $0\leq \deg G\leq n-2$ we have $k=1+\deg G$ and $d=n-\deg G$. Therefore $k+d=n+1$ hence $C$ is an MDS code. It follows that the weight enumerator of any rational AG code is known explicitly. Rational AG codes are described explicitly in  \cite[Sec 2.3]{st}. They are known as Generalized Reed Solomon codes.

\subsubsection{Elliptic AG codes} 

Let $C=C_\codeL(D,G)$ be an $[n,k,d]_q$-AG code of genus $g=1$. It follows from Weil-Serre estimations for the number of rational points that the maximal length $n$ of elliptic codes is $q+1\leq n\leq q+1+[2\sqrt q]$. Since $n+1-g\leq k+d\leq n+1$, either $d=n-k$, or $d=n-k+1$. 
\noindent 

(a) $[n,k,n-k+1]_q$-elliptic codes. These codes are MDS, and as we have seen before, their weight enumerators are known explicitly. 

(b) $[n,k,n-k]_q$-elliptic codes. The dual of an elliptic code is also an elliptic code, and the dual of an MDS code is also an MDS code. It follows that if $C$ is an elliptic $[n,k,n-k]_q$-code, then $C^{\perp}$ is an elliptic $[n,n-k,k]_q$-code. 
    
\begin{prop}
Let $C$ be a $[n,k,n-k]_q$-elliptic code with weight enumerator 
\[ \displaystyle{A_C(x,y)=x^n+\sum_{i=n-k}^n A_ix^{n-i}y^i}.\]
Let 
\[ \displaystyle{A_{C^\perp}  (x,y)=x^n+\sum_{i=k}^n A^{\perp}_ix^{n-i}y^i},\]
be the weight enumerator of $C^{\perp}$. 
\begin{enumerate}
\item $A_C(x,y)$ is completely determined by $A_{n-k}$ as follows 
\[ A_{n-k+l}={n \choose k-l}\sum_{i=0}^{l-1}(-1)^i{n-k+l \choose i}(q^{l-i}-1)+(-1)^l{k \choose k-l}A_{n-k},\]
for all $0\leq l\leq k$.

\item $A^{\perp}_k=A_{n-k}$, i.e. $C$ and $C^{\perp}$ have the same number of minimum weight codewords.

\item If $n=2k$ then $A_C(x,y)=A_{C^{\perp}}(x,y)$, i.e. $C$ and $C^{\perp}$ are formally self-dual.
\end{enumerate}
\end{prop}
 
\begin{proof}
Recall from Prop~\ref{dual weights} or Eq.~\eqref{eq-8} the MacWilliams relation 
\begin{equation}
\sum_{i=0}^{n-l}{n-i \choose l}A_i=q^{k-l}\sum_{i=0}^l{n-i \choose n-l}A_i^{\perp},~0\leq l\leq n.
\end{equation}
But $A_0=A^{\perp}_0=1$ and $A_1=\dots =A_{n-k-1}=A^{\perp}_1= \dots =A^{\perp}_{k-1}=0$, which for $l=k$ yield (2). Otherwise, we get get 
\[ \sum_{i=n-k}^{n-l}{n-i \choose l}A_i=q^{k-l}{n \choose l}(q^{k-l}-1),~ l=0,1, \dots ,k-1.\]
If $A_{n-k}$ is known, then the equation with $l=k-1$ yields $A_{n-k+1}$, the equation with $l=k-2$ yields $A_{n-k+2}$ and so on. All numbers $A_j$ can be found. The formula for $A_{n-k+l}$ in (1) is found by using induction. Statement (3) follows easily from (1) and (2).
\end{proof}

Thus, for the computation of the weight distribution of an $[n,k,n-k]_q$-elliptic code it is sufficient to compute the number of minimum weight codewords $A_{n-k}$. We present this number in a few special cases when the elliptic code is of maximal length $n=|\X(\F_q)|$, the general case is much more complicated and will not be addressed here. 

\begin{prop}
Let $\X (\F_q)=\{P_1,P_2,...,P_n\}$, and $D=P_1+P_2+...+P_n$. Let $G$ be a divisor such that $\X (\F_q)\cap \supp~G=\emptyset,~0<\deg G<n,$ and $C_\codeL(D,G)$ is not MDS. If $k=\deg G$ and $n$ are co-prime, then $$A_{n-k}=\frac{q-1}{n}{n \choose k}.$$
\end{prop}

For the second result, recall that there is a bijection from the Jacobian of $\X$ onto the set of rational points $\X (\F_q)$.  This gives rise to an algebraic operation $\oplus$ on $\X(\F_q)$ such that $(\X (\F_q),\oplus)$ is an abelian group. Denote the identity of this group by $P$. 
Let $\X (\F_q)=\{P,P_1,P_2,...,P_n\}$, and $D=P_1+P_2+...+P_n$.

\begin{prop}
Let $0< k< n$ and $G:=kP_1$. Assume that $k!$ and $n+1=|\X(\F_q)|$ are coprime, and that $C_\codeL(D,G)$ is not MDS. Then 
\[ A_{n-k}=\frac{q-1}{n+1}\left[{n \choose k}+(-1)^kn\right] \]
\end{prop}

The proofs of these last two results can be found in \cite{shokrollahi}.

\subsubsection{Higher genus AG codes} Let $C=C_\codeL(D,G)$ be an $[n,k,d]_q$-AG code of genus $g\geq 2$. Denote $m=:\deg G$. Since  $d\geq n-m$, we can re-write the weight enumerator of $C$ as follows: 
\[A_C(x,y)=x^n+\sum_{i=0}^m A_{n-i}x^iy^{n-i}=x^n+\sum_{l=0}^m B_l(x-y)^l,\]
where 
\[B_l=\sum_{i=n-m}^{n-l}{n-i \choose m}A_i\geq 0.\]
Using MacWilliams identity (Prop~\ref{dual weights}), we get 

\begin{thm}(Theorem \ref{MacWilliams solutions}, \cite{D6},\cite{TVT})
Let $C$ be an AG code of genus $g$. Then for $0\leq l\leq m-2g+1$ we have $$B_l={n \choose l}(q^{m-l-g+1}-1),$$ and for $m-2g+2\leq l\leq m$ we have 
\[\max \left\{0,{n \choose l}(q^{m-l-g+1}-1)\right\}\leq B_l\leq {n \choose l}(q^{\floor{(m-l)/2}+1}-1).\]
\end{thm}
  
\noindent Thus, there are $2g-1$ unknown parameters $B_l,~m-2g+2\leq l\leq m$ in the weight enumerator of an AG code of genus $g$.

\begin{prob}
Compute the parameters $B_l,~m-2g+2\leq l\leq m$ in the case of hyper-elliptic or super-elliptic curves.
\end{prob}


\bibliographystyle{plain}

\bibliography{mybib}{}


\end{document}